\documentclass[12pt]{amsart}

\usepackage{color}

\usepackage{amsmath,amsfonts,amsbsy}
\usepackage{eucal}
\usepackage{dsfont}

\usepackage{xy}
\xyoption{matrix} \xyoption{arrow} \xyoption{arc} \xyoption{color}
\usepackage{subfigure}
\usepackage{tikz}
\usetikzlibrary{arrows}
\usepackage{yhmath}

\usepackage{enumerate}
\usepackage{enumitem}
\makeatletter
\def\namedlabel#1#2{\begingroup
    #2%
    \def\@currentlabel{#2}%
    \phantomsection\label{#1}\endgroup
}
\makeatother

\usepackage[latin1]{inputenc}  %%%%%% French and italian accented characters

\numberwithin{equation}{section}

\makeatletter
\newtheoremstyle{corsivo}
   {\medskipamount}{\medskipamount}%
   {\itshape}{}%
   {\bfseries}{}%
   { }
   {\thmname{#1}\thmnumber{\@ifnotempty{#1}{ }\@upn{#2}}%
    \thmnote{ {\bfseries(#3)}}.}%
\makeatother

\theoremstyle{corsivo}
\newtheorem{thm}{Theorem}%[section]

\newtheorem{prop}{Proposition}

\makeatletter
\newtheoremstyle{dritto}
   {\medskipamount}{\medskipamount}%
   {\rmfamily}{}%
   {\bfseries}{}%
   { }
   {\thmname{#1}\thmnumber{\@ifnotempty{#1}{ }\@upn{#2}}%
    \thmnote{ {\bfseries(#3)}}.}%
\makeatother

\theoremstyle{dritto}
\newtheorem{dfn}{Definition}
\newtheorem{rmk}{Remark}

\newtheorem{assumption}{Assumption}

\newcommand{\eps}{\varepsilon}

\newcommand{\Id}{\mathds{1}}
\newcommand{\id}{\mathbb{I}}
\newcommand{\B}{\mathbb{B}}
\newcommand{\Bred}{\mathbb{B}\sub{eff}}
\newcommand{\C}{\mathbb{C}}
\newcommand{\R}{\mathbb{R}}
\newcommand{\Z}{\mathbb{Z}}
\newcommand{\N}{\mathbb{N}}
\newcommand{\T}{\mathbb{T}}
\newcommand{\PB}{\mathcal{P}}
\newcommand{\Hi}{\mathcal{H}}
\newcommand{\Hf}{\mathcal{H}\sub{f}}
\newcommand{\U}{\mathcal{U}}
\newcommand{\UZ}{\mathcal{U}\sub{BF}}

\newcommand{\BH}{\mathcal{B}(\mathcal{H})}

\newcommand{\scal}[2]{\left\langle #1, #2 \right\rangle}
\newcommand{\norm}[1]{\left\| #1 \right\|}

\newcommand{\bra}[1]{\left\langle #1 \right|}
\newcommand{\ket}[1]{\left| #1 \right\rangle}
\newcommand{\eu}{\mathrm{e}}
\newcommand{\iu}{\mathrm{i}}
\newcommand{\di}{\mathrm{d}}

\newcommand{\sub}[1]{_{\mathrm{#1}}}

\DeclareMathOperator{\tr}{tr}
\DeclareMathOperator{\Ran}{Ran}

\DeclareMathOperator{\Span}{Span}

\newcommand{\ie}{{\sl i.\,e.\ }}
\newcommand{\eg}{{\sl e.\,g.\ }}
\newcommand{\set}[1]{ \left\{  #1 \right\}}

\newcommand{\virg}[1]{``#1''}
\newcommand{\half}{\mbox{\footnotesize $\frac{1}{2}$}}

\newcommand{\la}{\lambda}

\def\({\left(}
\def\){\right)}

\let\oldfootnote\footnote
\renewcommand{\footnote}[1]{\oldfootnote{\  #1}}

\title[Symmetry and localization in periodic crystals]{Symmetry and localization in periodic crystals: \\[1.5mm] 
triviality of Bloch bundles with \\[1.5mm]  a fermionic time-reversal symmetry 
\vspace{0mm}}
\author{Domenico Monaco \and Gianluca Panati \vspace{0mm} }
\date{October 25, 2014. Final version, revised according to the remarks by the Reviewers.}

\usepackage%[pdftex,%
{hyperref}

\begin{document}

\begin{abstract}  We describe some applications of group- and bundle-theoretic methods in solid state physics, 
showing how symmetries lead to a proof of  the localization of electrons in gapped crystalline solids, as \eg 
 insulators and semiconductors.  

\noindent We shortly review the Bloch-Floquet decomposition of periodic operators,  and the related concepts of Bloch frames and composite Wannier functions. We show that the latter are almost-exponentially localized if and only if there exists a smooth periodic Bloch frame, and that the obstruction to the latter condition is the triviality of a Hermitian vector bundle, called the \emph{Bloch bundle}.  
The r\^ole of additional $\Z_2$-symmetries,  as time-reversal and space-reflection symmetry, is discussed, 
showing how  time-reversal symmetry implies the   triviality of the Bloch bundle, both in the bosonic and in the fermionic case. 
Moreover, the same $\Z_2$-symmetry allows to define a finer notion of isomorphism and, consequently, to define new topological invariants, which agree with the indices introduced by Fu, Kane and Mele in the context of topological insulators.   

\medskip

\noindent \textsc{Keywords:}  Periodic Schr\"{o}dinger operators, composite Wannier functions, Bloch bundle, 
Bloch frames, time-reversal symmetry, space-reflection symmetry, invariants of topological insulators. 

\smallskip

\noindent  \textsc{Note:} Contribution to the proceedings of the conference {``SPT2014 -- Symmetry and Perturbation Theory''}, Cala Gonone, Italy (2014).
\vspace{-10mm}
\end{abstract}

\maketitle

\tableofcontents

%%%%% SECTION 1  %%%%%%%%%%%%%%%%%%%%%%%%%%%%%%%%%%
\newpage

\section{Symmetries in solid state physics}

\begin{quote}
\virg{\it Symmetry, as wide or narrow as you may define its meaning, is one idea by which man through the ages has tried to comprehend and create order, beauty, and perfection.}  %\textsc{(H. Weyl, {\it Symmetry})}
\begin{flushright} \textsc{(H. Weyl, {\it Symmetry})} \end{flushright}
\end{quote}

\bigskip

Symmetries play a crucial r\^ole in the understanding of physical systems. Quantum systems are not exceptional, and since the dawn of Quantum Mechanics the archetypical idea of symmetry leaded researches to the solution of a wide range of problems, from  atomic to high energy physics.  The purpose of this contribution is to emphasize some relevant application of symmetries to solid state physics, and to provide a link with some new results on the geometric invariants of recently discovered crystalline solids, known as topological insulators \cite{HasanKane}. 

Most of the solids which appear to be homogeneous at the macroscopic scale are modeled by a Hamiltonian 
operator which is invariant with respect to translations by vectors in a Bravais lattice $\Gamma \simeq \Z^d$, the exceptions being confined to the pioneering field of aperiodic solids \cite{BellissardNerrmannZarrouati2000}.  As early realized, this $\Z^d$-symmetry can be used to decompose the problem via the so-called  Bloch-Floquet transform (Section \ref{Sec:Bloch-Floquet transform}). In gapped systems, two optional 
$\Z_2$-symmetries, namely the time-reversal (TR) symmetry  and the space-reflection (SR) symmetry,  reflect in special properties of the projector up to the gap, which are reviewed in Section \ref{Sec:Z_2 symmetries}.  

A crucial problem in the theory of periodic solids is to investigate the localization of the composite Wannier functions  
associated to a physically relevant family of Bloch bands (Section \ref{Sec:localization}). 
Indeed, the existence of exponentially-localized composite Wannier functions is a fundamental theoretical tool to derive tight-binding effective models, and to develop numerical algorithms whose computational cost scales only linearly with the size of the confining box \cite{BrouderPanati2007}.  Such existence problem can be shown to be equivalent to the triviality of a Hermitian vector bundle, called the \emph{Bloch bundle}. Time-reversal symmetry is crucial to prove the triviality of the latter. Here we review the proof in \cite{Panati} concerning systems with a \emph{bosonic} TR-symmetry, and we extend it to  the case of a \emph{fermionic} TR-symmetry.

In both cases the Bloch bundle is trivial, so there are no non-trivial topological indices which are invariant with respect to all the continuous deformations of the Hamiltonian which preserve both the gap and the $\Z^d$-symmetry. However, further topological information appears if one focuses on those continuous deformations of the Hamiltonian which respect also a fermionic TR-symmetry.  In such case, new $\Z_2$ invariants appear \cite{GrafPorta,FiMoPa2},  which can be proved to equal the indices introduced by Fu, Kane and Mele  \cite{FuKa,FuKaneMele} in the context of TR-symmetric topological insulators (Section \ref{Sec:Z_2 invariants}).

%%%%%%%%%%%%%%%%%%%%%%%%%%%%%%%%%%%%%%%%%%%%%%%%%%%%%%%%%
\subsection{Bloch-Floquet transform}
\label{Sec:Bloch-Floquet transform}

In solid state physics, one is interested in studying systems which have a \emph{$\Z^d$-symmetry}, given by the periodicity with respect to translations by vectors in the Bravais lattice $\Gamma = \Span_\Z \set{a_1, \ldots, a_d } \simeq \Z^d \subset \R^d$ of the solid under consideration. The Hamiltonian $H_\Gamma$ of the system is thus required to commute with these translation operators $T_\gamma$:
\begin{equation} \label{commuta}
[H_\Gamma, T_\gamma] = 0 \quad \text{for all} \quad \gamma \in \Gamma.
\end{equation}

For example, in continuous models the Hamiltonian and the translation operators act on the Hilbert space $\Hi := L^2(\R^d) \otimes \C^{2s+1}$, corresponding to a single spin-$s$ particle in $d$-dimensions. 
%\[ (H_\Gamma \psi)(x) := - \Delta \psi(x) + V_\Gamma(x) \psi(x), \quad (T_\gamma \psi)(x) := \psi(x-\gamma), \quad \psi \in \Hi. \]
 The translations act according to the natural prescription        
\begin{equation} \label{Translations}
(T_\gamma \psi)(x) := \psi(x-\gamma)  
\end{equation}
while the Hamiltonian operators, usually called  {\it magnetic Bloch Hamiltonian} ($s=0$)  
 and \emph{periodic Pauli Hamiltonian} ($s=\half$),  are 
\begin{equation}  \label{?Hamiltonians}
\begin{aligned} 
& H\sub{MB} =   \half \( - \iu \nabla_x  + A_{\Gamma}(x) \)^2 +  V_\Gamma(x)                                       &\qquad & \text{for } s=0,\\[1.5mm]
& H\sub{Pauli} =  \half \(  (- \iu \nabla_x  + A_{\Gamma}(x)) \cdot \sigma \)^2 +  V_\Gamma(x)            & \qquad & \text{for } s=\half,  \\
\end{aligned}
\end{equation}  
where $\sigma = (\sigma_1, \sigma_2, \sigma_3)$ is the vector consisting of the three Pauli matrices. 
All over the paper, we use Hartree atomic units, and moreover we reabsorb the reciprocal of the speed of light
$1/c$ in the definition of the function $A_{\Gamma}$.

The commutation relation \eqref{commuta} implies that the (rescaled) magnetic vector potential   $A_{\Gamma} : \R^d \to \R^d$  and the scalar potential  $V_\Gamma: \R^d \to \C^{2s +1}$ are $\Gamma$-periodic functions, which in particular implies that the magnetic flux per unit cell $\Phi_B$ is zero.    
The case of a non-zero magnetic flux per unit cell, which generically appears when \eg a \emph{uniform} magnetic field is considered, can be recasted in this framework provided that  the natural translations \eqref{Translations} are replaced by the \emph{magnetic translations}  \cite{Zak1964} and that  $\Phi_B $ is a rational multiple of the fundamental flux unit 
$\Phi_0 =  hc/e$.  The case $\Phi_B / \Phi_0  \notin 2 \pi \mathbb{Q}$  is instead radically different, and its  mathematical analysis  requires novel ideas and methods from non-commutative geometry \cite{Bellissard IQHE1994}. 

While both the Hamiltonians \eqref{?Hamiltonians} can be studied with the methods described in this Section, for the sake of simplicity we will mainly refer to the paradigmatic case of a periodic \emph{real}  Schr\"odinger operator, acting as 
\begin{equation} \label{Simple Hamiltonian}
(H_\Gamma \psi)(x) := - \half \Delta \psi(x) + V_\Gamma(x) \psi(x),
\end{equation}
where $V_{\Gamma}$ is a real-valued $\Gamma$-periodic function. In view of the commutation relation \eqref{commuta}, one may look for simultaneous eigenfunctions of $H_{\Gamma}$ and the translations $\set{T_{\gamma}}_{\gamma \in \Gamma}$, 
\ie for a solution to the problem
\begin{equation} 
\label{Simultaneous eigenf}
\begin{cases}
\(T_{\gamma} \psi \)(x)  = \omega_{\gamma} \, \psi(x)                 &\qquad  \omega_{\gamma} \in U(1),   \\
(- \half \Delta + V_{\Gamma}) \psi  =  E  \,\, \psi                            & \qquad E \in \R. 
\end{cases}
\end{equation}
The eigenvalues of the unitary operators $T_\gamma$  provide an irreducible representation 
$\omega \colon \Gamma \to U(1)$, $\gamma \mapsto \omega_\gamma$,   of the abelian group $\Gamma \simeq \Z^d$: it follows that $\omega_\gamma$ is a character, \ie
\[ \omega_\gamma = \omega_\gamma(k) = \eu^{\iu k \cdot \gamma}, \quad \mbox{ for some } k \in \T^d_* := \R^d / \Gamma^*. \]
Here $\Gamma^*$ denotes the dual lattice of $\Gamma$, given by those $\lambda \in \R^d$ such that $\lambda \cdot \gamma \in 2 \pi \Z$ for all $\gamma \in \Gamma$. The quantum number $k \in \T^d_*$ is called \emph{crystal} (or \emph{Bloch}) \emph{momentum}, and the quotient $\T^d_* =  \R^d / \Gamma^*$ is often called \emph{Brillouin torus}. 
Thus, the eigenvalue problem \eqref{Simultaneous eigenf}  reads   
\begin{equation} 
\label{Simultaneous eigenf_2}
\begin{cases}
\psi(k,x  - \gamma)  =  \eu^{\iu k \cdot \gamma} \, \psi(k,x)                 &\qquad k \in \T^d_*\\
(- \half \Delta + V_{\Gamma}) \psi(k, x)  =  E  \, \psi(k, x)                                 & \qquad E \in \R. 
\end{cases}
\end{equation}
In view of the first equation, a non-zero solution can not exist in $L^2(\R^d)$, so one looks 
for solutions in  $L\sub{loc}^2(\R^d)$. These solutions $\psi(k, \cdot)$ are called \emph{generalized eigenfunctions}  and normalized by imposing 
$
\int_{Y} | \psi(k,x)|^2 \di x = 1,   
$
where $Y$ is a fundamental unit cell for the lattice $\Gamma$. 

%In view of the commutation relation \eqref{commuta}, it makes sense to look for eigenfunctions of $H_\Gamma$ having a well-defined crystal momentum $k$. 

We denote the eigenvalues and eigenvectors of $H_\Gamma$ at fixed Bloch momentum as $E_n(k)$ and $\psi_n(k, \cdot)$, $n \in \N$, respectively. The functions $k \mapsto E_n(k)$ are called \emph{Bloch bands}, and the functions $k \mapsto \psi_n(k, \cdot)$ are called \emph{Bloch functions} in the physics literature. According to the so-called \emph{Bloch theorem} \cite{Kittel}, one can write
\begin{equation} \label{periodicBloch}
\psi_n(k,x) = \eu^{\iu k \cdot x} u_n(k,x)
\end{equation}
where $u_n(k, \cdot)$ is, for any fixed $k$,  a $\Gamma$-periodic function of $x$, thus living in the Hilbert space $\Hf := L^2(\T^d)$, with $\T^d = \R^d / \Gamma$.

\medskip

A more elegant and useful approach to obtain such Bloch functions (or rather their $\Gamma$-periodic part), is provided by adapting ideas from harmonic analysis. Indeed,  one can proceed in analogy with the free particle case, where the Fourier representation gives a way to diagonalize simultaneously both the Laplacian and the translations. Formally, one introduces the so-called (\emph{modified}) \emph{Bloch-Floquet transform},%%%
\footnote{A comparison with the \emph{classical} Bloch-Floquet transform, appearing in physics textbooks, is provided in Remark \ref{rmk:cBF}.}\  %%%%
 acting on functions $w \in C_0(\R^d) \subset L^2(\R^d)$ by%%%
\footnote{We intentionally use the symbol $k$, already appearing in \eqref{Simultaneous eigenf_2} and \eqref{periodicBloch} with an \emph{a priori} different meaning, {also in \eqref{Zak transform},} since it will be clear in few lines that the $k$ appearing in \eqref{Zak transform} can be naturally identified with the Bloch momentum introduced above (compare \eqref{eigenvalue}). 
} %%%End footnote
\begin{equation} \label{Zak transform}
( \UZ \, w)(k,y):= \frac{1}{|\B|^{1/2}} \,\, \sum_{\gamma\in\Gamma} \eu^{-\iu k \cdot (y - \gamma)} \, w( y - \gamma), \qquad y \in \R^d, \, k \in\R^{d}.
\end{equation}
Here $\B$ denotes the fundamental unit cell for $\Gamma^*$, namely
\[ \B := \set{k = \sum_{j=1}^{d} k_j b_j \in \R^d: -\frac{1}{2}  \le k_j \le \frac{1}{2} } \]
where the dual basis $\set{b_1, \ldots, b_d} \subset \R^d$, spanning $\Gamma^*$, is defined by $b_i \cdot a_j = 2 \pi \delta_{i,j}$.

Roughly speaking, the operator $\UZ$ separates the ``slow'' degrees of freedom, corresponding to $\gamma \in \Gamma$, from the ``fast'' degrees of freedom ($y$ in a unit cell for $\Gamma$), and can be interpreted as a discrete Fourier transform in the ``slow'' degrees of freedom alone. As such, one can expect $\UZ$ to be implemented unitarily on $L^2(\R^d)$. To determine the correct target Hilbert space, one first recognizes from the definition \eqref{Zak transform} that the function $\varphi(k,y) = (\UZ \, w)(k,y)$ is $\Gamma$-periodic in $y$ and $\Gamma^*$-pseudoperiodic in $k$, \ie
\[ \varphi(k + \lambda,y) = \left(\tau(\lambda) \varphi\right)(k,y) := \eu^{- \iu \lambda \cdot y} \varphi(k,y), \quad \lambda \in \Gamma^*. \]
The operators $\tau(\lambda) \in \U(\Hf)$ defined above provide a unitary representation of the group of translations by vectors in the dual lattice $\Gamma^*$. Following \cite{PST2003}, we define the Hilbert space of $\tau$-equivariant 
$L\sub{loc}^2$-functions as
\[ \Hi_\tau :=\Big\{ \varphi \in L^2_{\rm loc}(\R^d; \Hf ):\,\, \varphi(k + \lambda) = \tau(\lambda)\,\varphi(k) \; \forall \lambda \in \Gamma^{*}, \mbox{ for a.e. } k \in \R^d \Big\}\,. \]
Such functions are uniquely specified by the values they attain on the unit cell $\B$. One can hence identify $\Hi_\tau$ with the \emph{constant fibre direct integral} \cite[Sec. XIII.16]{Reed-Simon}
\begin{equation} \label{decomposition} 
\Hi_\tau \simeq L^2(\B; \Hf) \simeq \int_{\B}^{\oplus} \Hf \, \di k.
\end{equation}
Then the Bloch-Floquet transform $\UZ$ given in \eqref{Zak transform} can be extended to a unitary operator
\[ \UZ \colon L^2(\R^d) \to \Hi_\tau. \]

With respect to the decomposition \eqref{decomposition}, one has indeed a simultaneous ``diagonalization'' of periodic differential operators and translation operators, in the sense that
\begin{align*}
\UZ \, T_\gamma \, \UZ^{-1} & = \int_{\B}^{\oplus} \left( \eu^{\iu k \cdot \gamma} \Id \right) \, \di k, &\\
\UZ \left( - \iu \, \frac{\partial}{\partial x_j} \right) \UZ^{-1} & = \int_{\B}^{\oplus} \left( - \iu \, \frac{\partial}{\partial y_j} + k_j \right) \, \di k,    
& j \in \set{1, \ldots, d}, \\
\UZ f_\Gamma(x) \UZ^{-1} & = \int_{\B}^{\oplus} f_\Gamma(y) \, \di k, 
&    \mbox{if $f_\Gamma$ is $\Gamma$-periodic}.
%f_\Gamma(x - \gamma) = f_\Gamma(x)  \forall \gamma \in \Gamma. 
% for $f_\Gamma$ a $\Gamma$-periodic function.
\end{align*}
In particular, in the Bloch-Floquet representation the Hamiltonian $H_\Gamma = - \half \Delta + V_\Gamma$ becomes the fibred operator
\[ \UZ \, H_{\Gamma} \, \UZ^{-1} = \int_{\B}^\oplus H(k)\,\di k , \quad \text{where} \quad H(k) = \half \big( -\iu \nabla_y + k\big)^2 + V_\Gamma(y). \]
Whenever the operator $V_{\Gamma}$ is Kato-small with respect to the Laplacian (\ie infinitesimally $\Delta$-bounded), the operator $H(k)$ is self-adjoint on the $k$-independent domain $\mathcal{D}  = H^2(\T^d) \subset \Hf$. The $k$-independence of the domain of self-adjointness, which considerably simplifies the mathematical analysis, is the main motivation to use the \emph{modified} Bloch-Floquet transform \eqref{Zak transform}  instead of the classical one (compare Remark \ref{rmk:cBF}).  

The periodic part of the Bloch functions, appearing in \eqref{periodicBloch}, can be determined as solutions to the eigenvalue problem
\begin{equation} \label{eigenvalue}
H(k) u_n(k) = E_n(k) u_n(k), \quad u_n(k) \in \mathcal{D} \subset \Hf, \quad \norm{u_n(k)}_{\Hf} = 1.
\end{equation}
Even if the eigenvalue $E_n(k)$ has multiplicity $1$,  the eigenfunction $u_n(k)$ is not unique, 
since another eigenfunction can be obtained by setting 
$$
\widetilde u_n(k,y) = \eu^{\iu \theta(k)} u_n(k,y)
$$   
where $\theta: \T^d \to \R$ is any measurable function. We refer to this fact as the \emph{Bloch gauge freedom}. 

\medskip

In real solids, Bloch bands intersect each other. However, in insulators and semi-conductors the Fermi energy 
lies in a spectral gap, separating the occupied Bloch bands from the others. In this situation, it is convenient 
\cite{Bl, Cl2}  to regard all the bands below the gap as a whole, and to set up a multi-band theory. 
More generally, we select a portion of the spectrum of $H(k)$ consisting of a set of $m \ge 1$ physically relevant Bloch bands:
\[ \sigma_*(k) := \set{E_n(k) : n \in \mathcal{I}_* = \set{n_0, \ldots, n_0 + m - 1}}. \]
We assume that this set satisfies a \emph{gap condition}, stating that it is separated from the rest of the spectrum of $H(k)$, namely
\begin{equation}\label{Gap condition}
\inf_{k \in \B} \mathrm{dist}\big( \sigma_*(k), \sigma(H(k)) \setminus \sigma_*(k) \big) > 0.
\end{equation}
Under this assumption, one can define the spectral eigenprojector on $\sigma_*(k)$ as
\[ P_*(k) := \chi_{\sigma_*(k)}(H(k)) = \sum_{n \in \mathcal{I}_*} \ket{u_n(k, \cdot)} \bra{u_n(k, \cdot)}. \]
The equivalent expression for $P_*(k)$, given by the Riesz formula
\begin{equation} \label{Riesz formula }
P_*(k) = \frac{1}{2 \pi \iu} \oint_\mathcal{C} \left(H(k) - z \Id\right)^{-1} \, \di z, 
\end{equation}
where $\mathcal{C}$ is any contour in the complex plane winding once around the set $\sigma_*(k)$ and enclosing no other point in $\sigma(H(k))$, 
allows one to prove \cite[Prop.\ 2.1]{PanatiPisante} the following

\begin{prop} \label{P properties}
Let $P_*(k) \in \mathcal{B}(\Hf)$ be the spectral projector of $H(k)$ corresponding to the set $\sigma_*(k) \subset \R$. Assume that $\sigma_{*}$ satisfies the gap condition \eqref{Gap condition}. Then the family $\set{P_*(k)}_{k \in \R^d}$ has the following properties:
\begin{enumerate}[label=$(\mathrm{p}_{\arabic*})$, ref=$(\mathrm{p}_{\arabic*})$]
\item \label{p1} the map $k \mapsto P_*(k)$ is smooth from $\R^d$ to $\mathcal{B}(\Hf)$ (equipped with the operator norm);
\item \label{p2} the map $k \mapsto P_*(k)$ is $\tau$-covariant, \ie
\[ P_*(k + \lambda) = \tau(\lambda)^{-1} \, P_*(k) \, \tau(\lambda)  \qquad \forall k \in \R^d, \quad \forall \lambda \in \Gamma^{*}. \]
\end{enumerate}
\end{prop}

%If moreover the Hamiltonian is \emph{time-reversal symmetric}, that is, there exists an antiunitary operator%%%%%%
%\footnote{By \emph{antiunitary} operator we mean a surjective antilinear operator $C: \Hi \rightarrow \Hi$, such that $\scal{C\phi}{C \psi}_{\Hi} = \scal{\psi}{\phi}_{\Hi}$ for any $\phi, \psi \in \Hi$.} %%%%%%%%%%
%$T \colon \Hi \to \Hi$ such that $[H_\Gamma, T]=0$ and $T^2 = \pm \Id_\Hi$, then the family $\set{P_*(k)}_{k \in \R^d}$ has the further property
%\begin{enumerate}[resume*]
%\item \label{p3} there exists an antiunitary operator $\Theta\sub{f}$ acting on $\Hf$ such that
%\[ P_*(-k) =  \Theta\sub{f} \, P_*(k) \, \Theta\sub{f}^{-1}  \qquad \mbox{ and  } \qquad \Theta\sub{f}^{2} = \pm 1. \]
%\end{enumerate}

%%%%% Remark: comparison with classical Bloch-Floquet transform  %%%%%%%%%%%%%%%%%%%%%%

\medskip

\begin{rmk}[Comparison with classical Bloch-Floquet theory] \label{rmk:cBF}
In most solid state physics textbooks, the \emph{classical} Bloch-Floquet transform is defined as
\begin{equation} \label{cBF} 
( \U\sub{cl} \, w)(k,y):= \frac{1}{|\B|^{1/2}} \,\, \sum_{\gamma\in\Gamma} \eu^{ \iu k \cdot \gamma} \, w( y - \gamma), \qquad y \in \R^d, \, k \in\R^{d}
\end{equation}
for $w \in C_0(\R^d) \subset L^2(\R^d)$. The close relation with Fourier transform is thus more explicit in this formulation, and indeed the function $\varphi\sub{cl}(k,y) := ( \U\sub{cl} \, w)(k,y)$ will be $\Gamma^*$-periodic in $k$ and $\Gamma$-pseudoperiodic in $y$:
\begin{align*}
\varphi\sub{cl}(k+\lambda,y) & = \varphi\sub{cl}(k,y), \quad \lambda \in \Gamma^*, \\
\varphi\sub{cl}(k,y + \gamma) & = \eu^{ \iu k \cdot \gamma} \varphi\sub{cl}(k,y), \quad \gamma \in \Gamma.
\end{align*}

As is the case for the modified Bloch-Floquet transform \eqref{Zak transform}, the definition \eqref{cBF} extends to a unitary operator
\[ \U\sub{cl} \colon L^2(\R^d) \to \int_{\B}^{\oplus} \Hi_k \, \di k \]
where
\[ \Hi_k := \set{\varphi \in L^2\sub{loc}(\R^d) : \varphi(y+\gamma) = \eu^{\iu k \cdot \gamma} \varphi(y) \:\: \forall \, \gamma \in \Gamma, \text{ for a.e. } y \in \R^d}. \]
Moreover, a periodic Schr\"odinger operator of the form $H_\Gamma = - \half \Delta + V_\Gamma$ becomes, in the classical Bloch-Floquet representation,
\[ \U\sub{cl} \, H_{\Gamma} \, \U\sub{cl}^{-1} = \int_{\B}^\oplus H\sub{cl}(k)\,\di k , \quad \text{where} \quad 
H\sub{cl}(k) = - \half \Delta_y + V_\Gamma(y). \]
Although the form of the operator $H\sub{cl}(k)$, whose eigenfunctions $\psi_n(k,\cdot)$ appear in \eqref{periodicBloch}, looks simpler than the one of the fibre Hamiltonian $H(k)$ appearing in \eqref{eigenvalue}, one should observe that $H\sub{cl}(k)$ acts on a \emph{$k$-dependent} domain in the  \emph{$k$-dependent} Hilbert space $\Hi_k$. This constitutes the main disadvantage of working with the classical Bloch-Floquet transform \eqref{cBF}, thus explaining  why the modified definition \eqref{Zak transform} is preferred in the mathematical literature. 

The two Bloch-Floquet representations (classical and modified) are nonetheless equivalent, since they are unitarily related by the operator
\[ \mathcal{J} = \int_{\B}^{\oplus} \mathcal{J}_k \, \di k, \quad \text{where} \quad \mathcal{J}_k \colon \Hf \to \Hi_k, \quad \left(\mathcal{J}_k \varphi\right)(y) = \eu^{\iu k \cdot y} \varphi(y), \quad k \in \R^d, \]
see \eqref{periodicBloch}, so that in particular $\mathcal{J}_k H(k) \mathcal{J}_k^{-1} = H\sub{cl}(k)$. As a consequence, the Bloch bands $E_n(k)$ are independent of the chosen definition.  
\hfill $\Diamond$
\end{rmk}

%%%%%%%%%%%%%%%%%%%%%%%%%%%%%%%%%%%%%%%%%%%%%%%%%%%%%%%%%

\subsection{Localization of electrons and Bloch frames}
\label{Sec:localization}

Bloch functions can be used to study the properties of localization of electrons in the solid modeled by the Hamiltonian $H_\Gamma$. Indeed, one defines the associated \emph{Wannier functions} by going back to the position-space representation. More precisely, assume that $\sigma_*$ consist of a single isolated Bloch band $E_n$ (\ie $m=1$); the Wannier function $w_n$ associated to a choice of the Bloch function $u_n(k,\cdot)$ for the band $E_n$, as in \eqref{eigenvalue}, is defined by setting
\begin{equation} \label{Wannier}
w_n(x) := \left( \UZ^{-1} u_n \right)(x) = \frac{1}{|\B|^{1/2}} \int_{\B} \eu^{ \iu k \cdot x} u_n(k, x)\, \di k.
\end{equation}

In the multiband case ($m > 1$), one has to relax the notion of Bloch function to that of \emph{quasi-Bloch function} \cite{Cl2}, defined as an element $\phi \in \Hi_{\tau}$ such that
\[ P_*(k) \phi(k) = \phi(k), \qquad \norm{\phi(k)}_{\Hf}=1, \qquad \text{ for a.e. }  k \in \B. \]
A \emph{Bloch frame} is, by definition, a family of quasi-Bloch functions $\set{\phi_a}_{a=1, \ldots, m}$ which are orthonormal and span the vector space $\Ran P_*(k)$ at a.e.\ $k \in \B$. The \emph{composite Wannier functions}  $\set{w_1, \ldots, w_m} \subset L^2(\R^d)$ associated to a Bloch frame $\set{\phi_1, \ldots, \phi_m} \subset \Hi_{\tau}$ are defined in analogy with \eqref{Wannier} as 
\[ w_a(x) := \left( \UZ^{-1} \phi_a \right)(x) = \frac{1}{|\B|^{1/2}} \int_{\B} \eu^{ \iu k \cdot x} \phi_a(k, x)\, \di k. \]

Composite Wannier functions have become a standard tool in the analysis of localization properties of electrons in crystals 
\cite{MaVa, Wannier review}, by looking at their decay rate at infinity. One says that a set of composite Wannier functions is \emph{almost-exponentially localized} if
\[ \int_{\R^d} \left( 1 + |x|^2 \right)^r |w_a(x)|^2 \, \di x < \infty \quad \text{for all } r \in \N, \quad a \in \set{1, \ldots, m}. \]
If we denote by $X = (X_1, \ldots, X_d)$ the position operator in $L^2(\R^d)$, defined by $(X_j w)(x) := x_j w(x)$ on the maximal domain, then one has that in the Bloch-Floquet representation
\[ \UZ \, X \, \UZ^{-1} =  \iu \nabla_k .\]
In view of this, one can easily show \cite{PanatiPisante} that
\[ \int_{\R^d} \left( 1 + |x|^2 \right)^r |w_a(x)|^2 \, \di x < \infty \quad  \forall r \in \N \quad  
\Longleftrightarrow 
\quad \phi_a \in C^{\infty}(\R^d; \Hf) \cap \Hi_\tau. \]
Thus, the question of existence of almost-exponentially localized composite Wannier functions is reduced to the following

\medskip
\noindent \textbf{Question (Q)}: does there exist a \emph{smooth} Bloch frame $\set{\phi_a}_{a=1, \ldots, m} \subset \Hi_\tau$?
\medskip

As was noticed by several authors \cite{Kohn59,Cl2,Ne91}, there might be a competition between \emph{regularity} (a local issue) and \emph{periodicity} (a global issue) for a Bloch frame. Indeed, in general the above question might have 
a  negative answer due to a \emph{topological obstruction}, as we are going to illustrate in the next Section. 
In agreement with the vision of H. Weyl, symmetries play a fundamental r\^ole in the solution of this problem. Indeed, we will show in Section \ref{Sec:trivial} that Question \textrm{(Q)} has a positive answer, provided $d \leq 3$, whenever the system enjoys an additional $\Z_2$-symmetry, namely time-reversal symmetry.  

%%%%%%%%%%%%%%%%%%%%%%%%%%%%%%%%%%%%%%%%%%%%%%%%%
\goodbreak

\subsection{The r\^ole of additional symmetries}
\label{Sec:Z_2 symmetries}

Time-reversal (TR) symmetry is a further $\Z_2$-symmetry of some quantum systems, encoded in an antiunitary%%%%%%
\footnote{By \emph{antiunitary} operator we mean a surjective antilinear operator $C: \Hi \rightarrow \Hi$, such that $\scal{C\phi}{C \psi}_{\Hi} = \scal{\psi}{\phi}_{\Hi}$ for any $\phi, \psi \in \Hi$.} %%%%%%%%%%
operator $T$ acting on the Hilbert space $\Hi$ of the system. The time-reversal symmetry operator $T$ is called \emph{bosonic} or \emph{fermionic} depending on whether $T^2 = + \Id_\Hi$ or $T^2 = - \Id_\Hi$, respectively%%%%
\footnote{Since time-reversal symmetry flips the arrow of time, it must not change the physical description of the system if it is applied twice. Hence $T$ gives a projective unitary representation of the group $\Z_2$ on the Hilbert space $\Hi$, and as such
$ T^2 = \eu^{\iu \theta} \Id_\Hi$.
By antiunitarity, it follows that 
\[ \eu^{\iu \theta} T = T^2 T = T^3 = T T^2 = T \eu^{\iu \theta} \Id_\Hi = \eu^{-\iu \theta} T \]
and hence $\eu^{\iu \theta} = \pm 1$.}%%%%%%%%
. This terminology is motivated by the fact that there are ``canonical'' time-reversal operators when $\Hi  = L^2(\R^d) \otimes \C^{2s+1}$ with $s=0$ and $s=1/2$: in the former case, $T$ is just complex conjugation $C$ on $L^2(\R^d)$ (and hence squares to the identity), while in the latter $T$ is implemented as $C \otimes \eu^{\iu \pi S_y}$ on $\Hi = L^2(\R^d) \otimes \C^2$ (squaring to $- \Id_\Hi$), where 
$S_y = \half \sigma_2$  and $\sigma_2 = \begin{pmatrix} 0 & -\iu \\ \iu & 0 \end{pmatrix}$ is the second Pauli matrix.%, in the basis in which it is purely imaginary.

The following Proposition is a straightforward generalization of a result in \cite[Prop. 2.1]{PanatiPisante}, where the case 
$s=0$  is proved. 

\begin{prop}[Time-reversal symmetry]  \label{P TR-properties}
Under the hypotheses of Proposition \ref{P properties}, assume that the Hamiltonian is \emph{time-reversal symmetric}, 
that is, the Hamiltonian $H_{\Gamma}$ commutes with the canonical TR-operator  $T \colon \Hi \to \Hi$ defined above, 
$T^2 = \pm \Id_\Hi$.   Then the family $\set{P_*(k)}_{k \in \R^d}$ has the following property:%%%
\footnote{The following properties are labeled as \ref{p3} and \ref{p4}, since they are the natural complement of properties 
\ref{p1} and \ref{p2} appearing in Proposition \ref{P properties}. 
}\  %%%%
\begin{enumerate}[label=$(\mathrm{p}_{3, \pm})$, ref=$(\mathrm{p}_{3,\pm})$]
\item \label{p3} there exists an antiunitary operator $\Theta\sub{f}$ acting on $\Hf$ such that
\[ P_*(-k) =  \Theta\sub{f} \, P_*(k) \, \Theta\sub{f}^{-1}  \qquad \mbox{ and  } \qquad \Theta\sub{f}^{2} = \pm \Id. \]
\end{enumerate}
Moreover, one has the following 
\begin{enumerate}[label=$(\mathrm{p}_{4})$, ref=$(\mathrm{p}_{4})$]
\item \label{p4}   compatibility property:  $\Theta\sub{f} \, \tau_{\lambda} =  \tau_{-\lambda} \, \Theta\sub{f} $ for all 
$\lambda \in \Lambda$.
\end{enumerate}
\end{prop}

%\begin{proof}[Sketch of the proof.]  \NB{Here a sketch of the proof, assuming known Proposition 1 and the Riesz construction}
%\end{proof}

A \emph{real} Schr\"odinger operator, as in \eqref{Simple Hamiltonian}, obviously commutes with complex conjugation, 
and thus enjoys TR-symmetry, with $\Theta\sub{f}$ given by the complex conjugation in $\Hf = L^2(\T^d)$. 
For a spin-\half\ particle, one gets instead $\Theta\sub{f} =  C \otimes \eu^{\iu \frac{\pi}{2} \sigma_2}$ acting on 
$\Hf = L^2(\T^d) \otimes \C^2$.  More general Hamiltonians might be considered, but we remark that when dealing with the Hamiltonians \eqref{?Hamiltonians}, a non-zero magnetic potential $A_{\Gamma}$, even if yielding zero magnetic flux per unit cell (\ie no macroscopic magnetic field), does generically break time-reversal symmetry.   

\begin{rmk}[The r\^ole of $k=0$] \label{Rem:role of zero momentum} The reader might be surprised by the fact that in property \ref{p3} the point  $k=0$ has a distinguished r\^ole, thus breaking the translation invariance of the momentum space. 
This fact may be easily explained by noticing that a translation $k \mapsto k + \alpha$ in momentum space corresponds, 
via Bloch-Floquet transform, to a change of electromagnetic gauge in position space. More formally, setting 
$$
\widehat T_{\alpha} \phi (k,y) := \phi(k - \alpha, y),  \qquad    \text{ for } \phi \in \Hi_{\tau},        
$$
one easily sees that  
$$
\UZ^{-1} \, \widehat T_{\alpha}  \, \UZ = W_{\alpha}  \qquad \text{ where } \qquad   \( W_{\alpha} \psi \)(x) =  \eu^{\iu \alpha \cdot x} \psi(x).
$$ 
The  unitary operator $W_{\alpha}$  implements in $L^2(\R^d)$ a change of electromagnetic gauge, 
%%  corresponding to  $g(x) = \alpha \cdot x$
so that the magnetic vector potential is changed from $A(\cdot)$ to $A(\cdot) + \alpha$,  $\alpha \in \R^d$. 
Thus, the orbit $\set{ W_{\alpha}  H  W_{\alpha}^{-1} :  \alpha \in \R^d}$ of a given Hamiltonian under the action of a subgroup of the group of electromagnetic gauge transformations corresponds, via Bloch-Floquet transform, to the orbit of the transformed Hamiltonian under the action of translations in momentum space, 
namely to the set  
$$
\set{ \widehat T_{\alpha} \, \UZ \, H \, \UZ^{-1} \, \widehat T_{\alpha}^{-1} :  \alpha \in \R^d}.
$$ 
Whenever a distinguished element in the former orbit is TR-symmetric, it selects a distinguished point in the latter orbit (which is also TR-symmetric with respect to the fibre time reversal operator $\Theta\sub{f}$), 
and thus a point  $k_0 \in \R^d$.  As for the Hamiltonian $H_{\Gamma}$, as in \eqref{Simple Hamiltonian}, such distinguished 
point is $k_0 =0$, whose special r\^ole is now clarified. 
\hfill $\Diamond$
\end{rmk}

%%%% Here SR-symmetry  %%%%%%%%%%%%%%%%%%%%%%%%%%%%%%%%%%%
\medskip

It is worth to investigate how periodic quantum systems behave with respect to a fundamental $\Z_2$-symmetry of space, 
namely space-reflection symmetry, represented in $\Hi = L^2(\R^d) \otimes \C^{2s+1}$ by the unitary operator $R$ defined by
$$
\(R\psi \)(x) = \psi(-x). 
$$
In general, this symmetry does not hold true in crystalline solids. However, some solids enjoy the property of 
being \textbf{centrosymmetric}, in the sense that there exists a distinguished point $x_0 \in \R^d$ such that
\begin{equation} \label{Centrosymmetry}
V_{\Gamma}(\rho_{x_0}(x)) =  V_{\Gamma}(x)   \qquad \forall x \in \R^d, 
\end{equation}
where $\rho_{x_0}$ is the reflection with respect to the point $x_0$. Notice that the latter property involves both $V_{\Gamma}$ 
and $\Gamma$,  not just the Bravais lattice $\Gamma$. Whenever \eqref{Centrosymmetry} holds true, the Hamiltonian $H_{\Gamma}$ commutes with $R_{x_0}$, where $$\(R_{x_0} \psi \)(x) = \psi( \rho_{x_0}(x)).$$ 

Without loss of generality, we may always choose the origin of coordinates so that $x_0=0$, obtaining the identification $R_{x_0} \equiv R$.  The fact that the Hamiltonian commutes with $R$ yields a non-trivial unitary equivalence between $H(k)$ and $H(-k)$, which is the key to prove the following result. 

\begin{prop}[Space-reflection symmetry] \label{P SR-properties}
Under the hypotheses of Proposition \ref{P properties}, assume that the Hamiltonian is \emph{centrosymmetric}, 
that is,  $H_{\Gamma}$ commutes with the unitary operator  $R \colon \Hi \to \Hi$ defined above. 
Then the family $\set{P_*(k)}_{k \in \R^d}$ has the following property:
\begin{enumerate}[label=$(\mathrm{p}_{5})$, ref=$(\mathrm{p}_{5})$]
\item \label{p5} there exists a  unitary operator $R\sub{f}$ acting on $\Hf$ such that
\[ P_*(-k) =  R\sub{f} \, P_*(k) \, R\sub{f}^{-1}  \qquad \mbox{ and  } \qquad R\sub{f}^{2} = \Id. \]
In particular, one has explicitly $\( R\sub{f}\psi \)(y) = \psi(-y)$ for all $\psi \in \Hf$.  
\end{enumerate}
\end{prop}

\begin{proof}[Sketch of the proof.] First, we compute the action of $R$ in Bloch-Floquet representation. For any compactly supported $\psi \in L^2(\R^d)$  one has     
\begin{align*} 
\( \UZ  R \, \psi \)(k,y)    & =   \frac{1}{|\B|^{1/2}} \,\, \sum_{\gamma\in\Gamma} \eu^{-\iu k \cdot (y - \gamma)} 
\, \(R \psi \)( y - \gamma) \,\, = \\
          & =     \frac{1}{|\B|^{1/2}} \,\, \sum_{\gamma\in\Gamma} \eu^{-\iu (- k) \cdot (- y + \gamma)} \, \psi( - y + \gamma) = \\ 
          &=    \( \UZ \, \psi \)(- k, - y).       
\end{align*}
Then, a standard density argument yields that  
$$ 
\(\UZ \, R \, \UZ^{-1} \, \phi \)(k,y) = \phi(-k, -y)  \qquad \forall \phi \in \Hi_{\tau}.
$$    

Since $[H_{\Gamma}, R] =0$, one obtains that in the Bloch-Floquet representation
$$
H(-k) =  R\sub{f} \, H(k) \, R\sub{f}^{-1},      \qquad k \in \R^d,
$$ 
which in particular implies the parity of the spectrum, \ie $\sigma\(H(-k)\) = \sigma\(H(k)\)$. 
Since the spectrum of $H(k)$ is pure point spectrum,  the operator  $H(k)$ is bounded from below 
and \emph{the eigenvalues are labeled in increasing order}, one gets $E_n(k) = E_n(-k)$.   

The projector $P_*(k)$ is characterized by the Riesz formula \eqref{Riesz formula }, where the integration contour encloses the set $\sigma_*(k)$ and no other point in the spectrum of $H(k)$. By the gap condition \eqref{Gap condition}, 
the integration contour $\mathcal{C}_{k_0}$ chosen at $k_0 \in \R^d$ can be used to evaluate $P_*(k)$ for $k$ in a sufficiently small neighborhood $O_{k_0}$of  $k_0$.  Moreover, since $E_n(k) = E_n(-k)$ for all $k \in \R^d$, the same integration contour used locally around the point $k_0$ can be conveniently used around the point $-k_0$. 
Thus, by making this convenient choice, one obtains that for every $k \in O_{k_0}$  
\begin{eqnarray*}
P_*(-k)  &=&  \frac{1}{2 \pi \iu} \int_{\mathcal{C}_{k_0}}  \left(H(- k) - z \Id\right)^{-1} \, \di z  =  \\
             &=&  \frac{1}{2 \pi \iu} \int_{\mathcal{C}_{k_0}}  \left( R\sub{f} \, H(k) \, R\sub{f}^{-1}  - z \Id\right)^{-1} \, \di z   
             = R\sub{f} \, P_*(k) \, R\sub{f}^{-1}.   
\end{eqnarray*}

\noindent By the arbitrarity of $k_0$, the claim is proved. 
\end{proof}

The breaking of SR-symmetry, namely property \ref{p5}, is a necessary condition to observe a non-zero 
piezoelectric current and macroscopic polarization in crystals, see  \cite{PanatiSparberTeufel2009} and 
references therein. 

%%%%% SECTION 2  %%%%%%%%%%%%%%%%%%%%%%%%%%%%%%%%%%%%%%

\newpage
\section{The Bloch bundle and its geometry} 
\label{Sec:Bloch bundle}

\begin{quote}
\virg{\it In these days the angel of topology and the devil of abstract algebra fight for the soul of every individual discipline of mathematics.}  \begin{flushright} \textsc{(H. Weyl, {\it Invariants})} \end{flushright} 
\end{quote}

\medskip

In this Section, inspired again by the words of H. Weyl,  we will argue that the topological obstruction to the existence of a smooth Bloch frame is encoded in a smooth Hermitian vector bundle, baptized \emph{Bloch bundle} in \cite{Panati}. We also show 
that TR-symmetry, either of bosonic or fermionic type, implies the triviality of the Bloch bundle, and thus an affirmative answer to Question (Q).  

Abstracting from the specific case of periodic Schr\"odinger operators, we consider a family of orthogonal projectors acting on a separable Hilbert space $\Hi$, satisfying the following

\begin{assumption} \label{proj}
The family of orthogonal projectors $\set{P(k)}_{k \in \R^d} \subset \mathcal{B}(\Hi)$ enjoys the following properties:
\begin{enumerate}
\item[\namedlabel{item:smooth}{(P$_1$)}] \emph{smoothness}: the map $\R^d \ni k \mapsto P(k) \in \BH$ is $C^\infty$-smooth;
\item[\namedlabel{item:tau}{(P$_2$)}] \emph{$\tau$-covariance}: the map $k \mapsto P(k)$ is covariant with respect to a unitary representation $\tau \colon \Lambda \to \U(\Hi)$ of a maximal lattice $\Lambda \simeq \Z^d \subset \R^d$ on the Hilbert space $\Hi$, \ie
\[ P(k+\lambda) = \tau(\lambda) P(k) \tau(\lambda)^{-1}, \quad \text{for all } k \in \R^d, \lambda \in \Lambda; \]
\item[\namedlabel{item:TRS}{(P$_{3,\pm}$)}] \emph{time-reversal symmetry}: the map $k \mapsto P(k)$ is time-reversal symmetric, \ie there exists an antiunitary
operator $\Theta \colon \Hi \to \Hi$, called the \emph{time-reversal operator}, such that 
\[ \Theta^2 = \pm \Id_{\Hi} \quad \text{and} \quad P(-k) = \Theta P(k) \Theta^{-1}. \] 
\end{enumerate}

Moreover, the unitary representation $\tau \colon \Lambda \to \U(\Hi)$ and the time-reversal operator $\Theta \colon \Hi \to \Hi$ satisfy the compatibility condition
\begin{equation} \label{item:TRtau} \tag{P$_4$} 
\Theta \, \tau(\lambda) = \tau(\lambda)^{-1} \, \Theta \quad \text{for all } \lambda \in \Lambda. 
\end{equation} \hfill  $\lozenge$
\end{assumption}

The previous Assumptions retain only the fundamental $\Z^d$- and $\Z_2$-symmetries of the family of eigenprojectors of a time-reversal symmetric periodic Hamiltonian, as in Propositions \ref{P properties} and \ref{P TR-properties}.  In this abstract framework, the analog of Question (Q) is the existence of a smooth $\tau$-equivariant 
Bloch frame, as in the following

\goodbreak

\begin{dfn}[Bloch frame] \label{dfn:Bloch}
Let $\mathcal{P} =\set{P(k)}_{k \in \R^d}$ be a family of projectors satisfying Assumptions \ref{item:smooth} 
and \ref{item:tau}.  A \textbf{local Bloch frame} for $\mathcal{P}$ on a region $\Omega \subset \R^d$ is a map 
\begin{eqnarray*}
\Phi : & \Omega & \longrightarrow \quad \Hi \oplus \ldots \oplus \Hi = \Hi^m \\
& k &\longmapsto  \quad  (\phi_1(k), \ldots, \phi_m(k))
\end{eqnarray*}  such that for a.e. $k \in \Omega$ the set $\set{\phi_1(k), \ldots, \phi_m(k)}$ is an orthonormal basis spanning $\Ran P(k)$.  If $\Omega = \R^d$ we say that $\Phi$ is a \textbf{global Bloch frame}.
Moreover, we say that a (global) Bloch frame is 
\begin{enumerate}[label=$(\mathrm{F}_{\arabic*})$,ref=$(\mathrm{F}_{\arabic*})$]
%\setcounter{enumi}{-1}
%\item  \emph{continuous} if the map $\phi_a : \R^d  \to \Hi^m$ is continuous for all $ a \in \set{1, \ldots, m}$;
\item  \label{item:F1}  \emph{smooth}  if the map $\phi_a : \R^d  \to \Hi^m$  is $C^{\infty}$-smooth for all $a \in \set{1, \ldots, m}$;
\item   \label{item:F2} \emph{$\tau$-equivariant} if 
\begin{equation*} \label{tau-cov}
\phi_a(k + \lambda) = \tau(\la) \, \phi_a(k) \quad \text{for all } k \in \R^d, \: \lambda \in \Lambda, \: a \in \set{1, \ldots, m}. \end{equation*} 
%\item \label{item:F3} \emph{time-reversal invariant} if
%\begin{equation*} \label{tr}
%\phi_a(-k) =  \Theta \, \phi_a(k)  \quad \text{for all } k \in \R^d, \: a \in \set{1, \ldots, m}.
%\end{equation*}
\end{enumerate}  \hfill  $\lozenge$ 
\end{dfn}

%%%%%%%%%%%%%%%%%%%%%%%%%%%%%%%%%%%%%%%%%%%%%%

Following \cite{Panati}, one can construct a Hermitian vector bundle $\mathcal{E}_\PB = \left(E_\PB \xrightarrow{\pi} \T^d_* \right)$, with $\T^d_* := \R^d / \Lambda$, called the \emph{Bloch bundle}, starting from a family of projectors $\PB := \set{P(k)}_{k \in \R^d}$ satisfying properties \ref{item:smooth} and \ref{item:tau}. One proceeds as follows: Introduce the following equivalence relation on the set $\R^d \times \Hi$:
\[ (k, \phi) \sim_\tau (k', \phi') \quad \text{if and only if} \quad \exists\, \lambda \in \Lambda : k' = k - \lambda \text{ and } \phi' = \tau(\lambda) \phi. \]
The total space of the Bloch bundle is then
\[ E_\PB := \set{ [k,\phi]_\tau \in (\R^d \times \Hi) / \sim_\tau \: : \: \phi \in \Ran P(k)} \]
with projection $\pi([k,\phi]_\tau) = k \pmod \Lambda \in \T^d_*$.  By using the Kato-Nagy formula \cite[Sec. I.4.6]{Kato}, 
one shows that the previous definition yields a smooth vector bundle, which moreover inherits  from $\Hi$ a natural Hermitian structure.   Question (Q) in the previous Section can be shown \cite[Prop. 2]{Panati} to be equivalent to 

\medskip
\noindent \textbf{Question ($\mathbf{Q^{\prime}}$)}: is the Bloch bundle $\mathcal{E}_\PB = \left(E_\PB \xrightarrow{\pi} \T^d_* \right)$ \emph{trivial} in the category of smooth Hermitian vector bundles over the torus $\T^d_*$? 
\medskip

%For the reader's convenience, 

We recall that a smooth vector bundle $\mathcal{E} = \left(E \xrightarrow{\pi} M \right)$ of rank $m$ is called \emph{trivial} if there is a smooth isomorphism to the product bundle $\mathcal{T} = \left(M \times \C^m \xrightarrow{\mathrm{pr}_1} M \right)$, where $\mathrm{pr}_1$ is the projection on the first factor.

\noindent If the Bloch bundle $\mathcal{E}_\PB$ is trivial, then a smooth Bloch frame $\set{\phi_a}_{a=1, \ldots, m}$ can be constructed by means of a smooth isomorphism $F : \T^d_* \times \C^m \xrightarrow{\sim} E_\PB$ by setting $\phi_a(k) := F(k,e_a)$, where $\set{e_a}_{a=1,\ldots, m}$ is any orthonormal basis in $\C^m$. 
Viceversa, a global smooth Bloch frame $\set{\phi_a}_{a=1, \ldots, m}$ provides a smooth isomorphism $G :  \T^d_* \times \C^m 
 \xrightarrow{\sim}  E_\PB$ by setting 
$$
G \left(k,(v_1, \ldots, v_m) \right) = [k, v_1 \phi_1(k) + \cdots + v_m \phi_m(k)]_{\tau}.
$$

In general, the triviality of vector bundles on a low-dimensional torus $\T^d_*$ with $d \le 3$ is measured by the vanishing of its \emph{first Chern class} \cite[Prop. 4]{Panati}, defined in terms of the family of projectors $\set{P(k)}_{k \in \R^d}$ by the formula
\begin{equation} \label{c1}
c_1(\mathcal{E}_\PB) = \frac{1}{2 \pi \iu} \sum_{1 \le \mu < \nu \le d} \Omega_{\mu \nu}(k) \di k_\mu \wedge \di k_\nu, \text{ with } \: \Omega_{\mu \nu}(k) = \tr_{\Hi} \left(P(k) \left[ \partial_\mu P(k), \partial_\nu P(k) \right] \right).
\end{equation}

Under the hypothesis \ref{item:TRS}, the Bloch bundle is equipped with a further structure, namely a fibre-wise antilinear morphism $\Theta_\PB \colon \mathcal{E}_\PB \to \mathcal{E}_\PB$ such that the following diagram commutes:
\[ \xymatrix{ E_\PB \ar[r]^{\Theta_\PB} \ar[d] & E_\PB \ar[d] \\ \T^d_* \ar[r]^{c} & \T^d_*} \]
where $c \colon \T^d_* \to \T^d_*$ denotes the involution $c(k) = -k$. This means that a vector in the fibre $\Ran P(k)$ is mapped via $\Theta_\PB$ into a vector in the fibre $\Ran P(-k)$. The morphism $\Theta_\PB \colon \mathcal{E}_\PB \to \mathcal{E}_\PB$ still satisfies $\Theta_\PB^2 = \pm \Id$, \ie it squares to the vertical automorphism of $\PB$ acting fibre-wise by multiplication by $\pm 1$. Following \cite{GrafPorta}, we call the previous structure a \emph{TR-symmetric Bloch bundle}.

The presence of this further $\Z_2$-symmetry is the key tool to provide a positive answer to Question ($\mathrm{Q}^{\prime}$), as we are now going to show.

%%%%%%%%%%%%%%%%%%%%%%%%%%%%%%%%%%%%%%%%%%%%%%%%%

\subsection{Triviality of TR-symmetric Bloch bundles: bosonic and fermionic cases}
\label{Sec:trivial}

The main result of \cite{Panati} is the following.

\begin{thm}[{\cite[Thm. 1]{Panati}}]
Let $d \le 3$, and let $\mathcal{P} = \set{P(k)}_{k \in \R^d}$ be a family of projectors satisfying properties \ref{item:smooth}, \ref{item:tau} and (P$_{3,+}$). Then the Bloch bundle $\mathcal{E}_\PB = \left(E_\PB \xrightarrow{\pi} \T^d_* \right)$ is trivial in the category of smooth Hermitian vector bundles.
\end{thm}

The above Theorem answers positively to Question ($\mathrm{Q}^{\prime}$), and hence to Question (Q), in the presence of \emph{bosonic} time-reversal symmetry: if the system enjoys time-reversal symmetry of bosonic type, then there exists a set of composite Wannier functions which are almost-exponentially localized. A natural question arises, namely whether the same kind of result holds also in presence of \emph{fermionic} time-reversal symmetry. This is proved in the following

\begin{thm} \label{thm:FermionicTrivial}
Let $d \le 3$, and let $\mathcal{P} = \set{P(k)}_{k \in \R^d}$ be a family of projectors satisfying properties \ref{item:smooth}, \ref{item:tau} and (P$_{3,-}$). Then the Bloch bundle $\mathcal{E}_\PB = \left(E_\PB \xrightarrow{\pi} \T^d_* \right)$ is trivial in the category of smooth Hermitian vector bundles.
\end{thm}
\begin{proof}
Arguing as in the proof of \cite[Thm. 1]{Panati}, the crucial point  is to show that, under hypothesis (P$_{3,-}$), the function
\[ \Omega_{\mu \nu}(k) = \tr_{\Hi} \left(P(k) \left[ \partial_\mu P(k), \partial_\nu P(k) \right] \right), \quad \mu, \nu \in \set{1, \ldots, d}, \]
which appears in the definition \eqref{c1} of the first Chern class, is \emph{odd} with respect to $k$, that is, $\Omega_{\mu \nu}(-k) = - \Omega_{\mu \nu}(k)$.

To prove this, we first observe that
\[ \Omega_{\mu \nu}(-k) = \tr_{\Hi} \left(P(-k) \left[ \partial_\mu P(-k), \partial_\nu P(-k) \right] \right). \]
Differentiating both sides of the equality $P(-k) = \Theta P(k) \Theta^{-1}$ (compare (P$_{3,-}$)) with respect to $k_\mu$, we obtain
\[ \partial_\mu P(-k) = - \Theta \partial_\mu P(k) \Theta^{-1} \]
and hence
\[ \Omega_{\mu \nu}(-k) = \tr_{\Hi} \left( \Theta P(k) \Theta^{-1} \left[ - \Theta \partial_\mu P(k) \Theta^{-1} , - \Theta \partial_\nu P(k) \Theta^{-1}  \right] \right). \]
%By bilinearity of the commutator, $[-A,-B] = [A,B]$, and also
%\[ [CAC^{-1}, CBC^{-1}] = CAC^{-1}CBC^{-1} - CBC^{-1}CAC^{-1} = C[A,B]C^{-1}, \]
%for $A,B,C \in \mathcal{B}(\Hi)$. 
The above expression simplifies to
\begin{align*}
\Omega_{\mu \nu}(-k)  %= \tr_{\Hi} \left( \Theta P(k) \Theta^{-1} \Theta \left[ \partial_\mu P(k), \partial_\nu P(k)\right] \Theta^{-1} \right) = \\
 = \tr_{\Hi} \left( \Theta P(k) \left[ \partial_\mu P(k), \partial_\nu P(k)\right] \Theta^{-1} \right).
\end{align*}

Since $\Theta$ is an \emph{anti}unitary operator, we have that $\tr_{\Hi}(\Theta A \Theta^{-1}) = \tr_{\Hi}(A^*)$ for any trace-class operator on $\Hi$. Indeed, by definition the trace of an operator $A$ is given by
\[ \tr_{\Hi}(A) = \sum_{n \in \N} \scal{\psi_n}{A \psi_n} \]
where $\set{\psi_n}_{n \in \N}$ is \emph{any} orthonormal basis of the Hilbert space $\Hi$.  
Noting now that $\Theta^{-1} = - \Theta$ (since $\Theta^2 = - \Id$), we can compute
\begin{align*}
\tr_{\Hi}(\Theta A \Theta^{-1}) & = - \tr_{\Hi}(\Theta A \Theta) = - \sum_{n \in \N} \scal{\psi_n}{\Theta A \Theta \psi_n} = \\
& = - \sum_{n \in \N} \scal{\Theta^2 A \Theta \psi_n}{\Theta \psi_n} = \sum_{n \in \N} \scal{A \Theta \psi_n}{\Theta \psi_n} = \\
& = \sum_{n \in \N} \scal{\Theta \psi_n}{A^* \Theta \psi_n} = \sum_{n \in \N} \scal{\widetilde{\psi}_n}{A^* \widetilde{\psi}_n} = \\
& = \tr_{\Hi}(A^*),
\end{align*}
where in the last equality we used the fact that $\widetilde{\psi}_n := \Theta \psi_n$ is just another orthonormal basis of $\Hi$, and the definition of the trace does not depend on the choice of the basis.

Hence
\begin{align*}
\Omega_{\mu \nu}(-k) & = \tr_{\Hi} \left\{ \left( P(k) \left[ \partial_\mu P(k), \partial_\nu P(k)\right] \right)^* \right\} = \\
& = \tr_{\Hi} \left( P(k)^* \left[ \partial_\mu P(k), \partial_\nu P(k)\right]^* \right).
\end{align*}
Observe now that $[A,B]^* = - [A^*, B^*]$ for $A,B \in \mathcal{B}(\Hi)$, and that $P(k)$ and its derivatives are \emph{self-adjoint} operators. This allows us to conclude finally that
\begin{equation} \label{Curvature odd}
\Omega_{\mu \nu}(-k) = - \tr_{\Hi} \left( P(k) \left[ \partial_\mu P(k), \partial_\nu P(k)\right] \right) = - \Omega_{\mu \nu}(k)
\end{equation}
as claimed.

The property of $\Omega_{\mu \nu}(k)$ of being odd implies that the first Chern class $c_1(\mathcal{E}_\PB)$ must vanish. Indeed, by de Rham duality it suffices to check that, when $c_1(\mathcal{E}_\PB)$ is integrated on any $2$-cycle
\[ \B_{\mu \nu} := \set{k \in \B: k_{\alpha} = 0 \text{ if } \alpha \ne \mu, \nu}, \quad 1 \le \mu < \nu \le d, \]
then the integral vanishes: this is because the $2$-cycles $\set{\B_{\mu \nu}}_{1 \le \mu < \nu \le d}$ generate the homology group $H_2(\T^d_*;\Z)$ of the $d$-torus. Divide $\B_{\mu \nu} = \B_{\mu \nu}^+ \cup \B_{\mu \nu}^-$, where $\B_{\mu \nu}^+$ (respectively $\B_{\mu \nu}^-$) contains only the points of $\B_{\mu \nu}$ with positive (respectively negative) $k_\mu$-coordinate. We have now
\begin{align*}
\int_{\B_{\mu \nu}} c_1(\mathcal{E}_\PB) & = \frac{1}{2 \pi \iu} \int_{\B_{\mu \nu}} \Omega_{\mu \nu}(k) \di k_\mu \wedge \di k_\nu = \\
& = \frac{1}{2 \pi \iu} \left( \int_{\B_{\mu \nu}^+} \Omega_{\mu \nu}(k) \di k_\mu \wedge \di k_\nu + \int_{\B_{\mu \nu}^-} \Omega_{\mu \nu}(k) \di k_\mu \wedge \di k_\nu \right) = \\
& = \frac{1}{2 \pi \iu} \left( \int_{\B_{\mu \nu}^+} \Omega_{\mu \nu}(k) \di k_\mu \wedge \di k_\nu + \int_{\B_{\mu \nu}^+} \Omega_{\mu \nu}(-k) \di k_\mu \wedge \di k_\nu  \right) = \\
& = \frac{1}{2 \pi \iu} \left( \int_{\B_{\mu \nu}^+} \Omega_{\mu \nu}(k) \di k_\mu \wedge \di k_\nu - \int_{\B_{\mu \nu}^+} \Omega_{\mu \nu}(k) \di k_\mu \wedge \di k_\nu \right) = 0,
\end{align*}
and this concludes the proof that $c_1(\mathcal{E}_\PB) = 0$.

By the technical lemma in \cite[Section 2.3]{Panati}, when $d \le 3$ the vanishing of the first Chern class is a necessary \emph{and sufficient} condition for the Bloch bundle $\mathcal{E}_\PB$ to be trivial as a smooth Hermitian vector bundle. This is equivalent to the existence of a global smooth $\tau$-equivariant Bloch frame for the family of projectors $\set{P(k)}_{k \in \R^d}$.
\end{proof}

\begin{rmk}[Exponentially localized Wannier functions]
The result of \cite{Panati} actually holds in the analytic category, if Assumption \ref{proj} is changed suitably to accommodate for the function $k \mapsto P(k)$ to be analytic in a strip $O_\alpha$  of width $2 \alpha$ around the ``real axis'' $\R^d \subset \C^d$. Under this modified assumption, one can indeed prove that if (P$_{3,+}$) holds then the corresponding Bloch bundle is trivial in the category of \emph{holomorphic} Hermitian vector bundles over the region $O_{\alpha} \supset \T^d_*$, by a general argument that goes under the name of \emph{Oka's principle} \cite[Chap. V]{Oka}. This is equivalent to the existence of composite Wannier functions which are \emph{exponentially localized}, namely
\[ \int_{\R^d} \eu^{2 \beta |x|} \left| w_a(x) \right|^2 \, \di x < \infty \quad \text{for all } a \in \set{1, \ldots, m}, \: 0 \le \beta < \alpha. \]

The same principle allows to extend our result (Theorem \ref{thm:FermionicTrivial}) to the analytic category also under hypothesis (P$_{3,-}$), and provides the existence of exponentially localized composite Wannier functions also in systems which have a time-reversal symmetry of fermionic type.
\hfill $\Diamond$
\end{rmk}

\begin{rmk}[Consequences of SR-symmetry]  Consider a gapped periodic system which enjoys both TR-symmetry, either of bosonic or fermionic type, and SR-symmetry.  As a consequence of \ref{item:TRS},  the Berry curvature of the Bloch bundle is odd, as showed by equation \eqref{Curvature odd}. On the other hand, arguing as in the proof of Theorem~\ref{thm:FermionicTrivial}, one shows that (the abstract analog of)  property \ref{p5} implies that the Berry curvature is even, namely  $\Omega_{\mu \nu}(-k) =  \Omega_{\mu \nu}(k)$, in view of the fact that $R\sub{f}$ is a \emph{unitary} operator. Thus, in systems which enjoy both TR- and SR-symmetry, \emph{the Berry curvature is identically zero}. This fact has some interesting consequences, since it implies that the parallel transport induced by the Berry connection is locally well-defined, \ie independent of the path chosen to connect the initial and the final point on the Brillouin torus $\T^d_*$. Notice, however,  that the holonomy induced by parallel transport might still be non-trivial, but only on paths which are not homotopic to the trivial path.  
\hfill $\Diamond$
\end{rmk}

%%%%%%%%%%%%%%%%%%%%%%%%%%%%%%%%%%%%%%%%%%%%%%%%%%%%%%%%%%%%%%%

\subsection{\texorpdfstring{$\Z_2$ invariants of fermionic TR-symmetric Bloch bundles}{Z2 invariants of fermionic TR symmetric Bloch bundles}}
\label{Sec:Z_2 invariants}

We have seen in the last Section that, with the help of the extra $\Z_2$-symmetry given by time-reversal (be it either bosonic or fermionic), we were able to ensure that the Bloch bundle is trivial. As was already observed before, this is equivalent to the existence of a Bloch frame which is both smooth and periodic, or rather $\tau$-equivariant, \ie satisfying \ref{item:F1} and \ref{item:F2}. The property of $\tau$-equivariance for a Bloch frame is clearly a compatibility of the frame itself with the corresponding $\Z^d$-simmetry of the family of projectors $\PB$, namely \ref{item:tau}. Once we have a family of projectors $\PB$ which satisfies also \ref{item:TRS}, is it possible to make the Bloch frame to be also compatible with the $\Z_2$-symmetry $\Theta$?

This issue is extremely relevant in the context of TR-symmetric topological insulators \cite{HasanKane}.  Indeed, 
as conjectured in \cite{FuKa,FuKaneMele}, a positive or negative answer to the previous question distinguishes between ordinary insulators and topological insulators, respectively.  While a rigorous mathematical approach to this issue has been first discussed in \cite{GrafPorta} for $2$-dimensional systems,  we focus here on a different method \cite{FiMoPa1, FiMoPa2} which has a natural generalization to $3$-dimensional systems.

A first remark is now in order, namely how one should formulate the above-mentioned compatibility condition. Looking at \ref{item:TRS}, the most natural choice would be to require that $\Phi(-k) = \Theta \Phi(k)$. However, if $k=0$ and the time-reversal operator $\Theta$ is of fermionic type, the bilinear form
\[ (\phi, \psi) := \scal{\Theta \phi}{\psi}, \quad \phi, \psi \in \Ran P(0) \]
is well-posed (notice that $\Ran P(0)$ is an invariant subspace for the action of $\Theta$ by \ref{item:TRS}) and defines a \emph{symplectic form} on $\Ran P(0)$. In particular, this implies that $\Ran P(0)$, and by continuity also $\Ran P(k)$ for all $k \in \R^d$, must be even-dimensional. If $\Phi(0) = \Theta \Phi(0)$ is a Bloch frame, each of its components $\phi$ satisfies $\scal{\phi}{\phi}=1$ by the normalization condition, and  
\[ \scal{\phi}{\phi} = \scal{\Theta \phi}{\phi} = (\phi, \phi) = 0 \]
by the skew-symmetry of the symplectic form $(\cdot, \cdot)$. Hence we are forced to require a more refined compatibility condition for a Bloch frame with time-reversal symmetry. 
Following the previous literature \cite{FuKa, FuKaneMele, GrafPorta},  we introduce a ``reshuffling matrix'' $\eps$ that exchanges the order of the entries of a Bloch frame $\Phi$ when going from $k$ to $-k$. This leads us to set the following

\medskip
\goodbreak 

\begin{dfn}[TR-symmetric Bloch frame]
Let $\PB = \set{P(k)}_{k \in \R^d}$ be a family of projectors satisfying \ref{item:smooth}, \ref{item:tau} and \ref{item:TRS}. A global Bloch frame $\Phi = \set{\phi_1, \ldots, \phi_m}$ for $\PB$ is said to be
\begin{enumerate} 
\item[\namedlabel{item:F3}{$(\mathrm{F}_{3})$}]  \emph{time-reversal symmetric} if
$$ %\begin{aligned}
\phi_b(-k) = \sum_{a=1}^{m} \Theta \phi_a(k) \eps_{ab}, 
$$ where 
$$
\eps =  % (\eps_{ab})_{1 \le a,b \le m} = 
\begin{cases} 
\id_m & \text{if $\PB$ satisfies (P$_{3,+}$),} \\    \\
\begin{pmatrix} 0 & \id_{m/2} \\ - \id_{m/2} & 0 \end{pmatrix} & \text{if $\PB$ satisfies (P$_{3,-}$).} 
\end{cases}
$$ %\end{aligned}
\end{enumerate} 
\hfill $\Diamond$
\end{dfn}

The natural issue that arises is now

\medskip
\noindent \textbf{Question (Q$_d$)}. Let $\PB = \set{P(k)}_{k \in \R^d}$ be a family of projectors satisfying \ref{item:smooth}, \ref{item:tau} and \ref{item:TRS}. Does there exist a global smooth Bloch frame which is both $\tau$-equivariant and time-reversal symmetric, \ie satisfying \ref{item:F1}, \ref{item:F2} and \ref{item:F3}?
\medskip

This Question is answered in \cite{FiMoPa1} for a bosonic TR-symmetry. As for a fermionic one, it is answered in 
\cite{GrafPorta} for $d=2$, and in \cite{FiMoPa2} for $d \leq 3$. In both approaches to the fermionic case, an obstruction appears, which for $d=2$ is encoded by a $\Z_2$-valued index, denoted by $\delta(\PB)$ in \cite{FiMoPa2}. Although the latter index is defined differently in \cite{GrafPorta} and \cite{FiMoPa2}, \emph{a posteriori} one proves that they agree, under suitable hypotheses. 
For $3$-dimensional systems, instead, the obstruction is measured by four  $\Z_2$-valued indices, denoted by
$\set{\delta_{1,0}(\PB), \delta_{j,+}(\PB)}_{j \in \set{1,2,3}}$ in \cite{FiMoPa2}.  While we refer to \cite{GrafPorta} for an overview of the method by Graf and Porta, we summarize our results in the following

%This Question is answered in \cite{FiMoPa1} for a bosonic TR-symmetry, and in \cite{FiMoPa2} for a fermionic one,  provided  
%$1 \le d \le 3$. The results are as follows: the quantities $\delta(\PB)$,  $\delta_{1,0}(\PB)$, $\delta_{1,+}(\PB)$, $\delta_{2,+}(\PB)$ and $\delta_{3,+}(\PB)$ which appear in the following statement are defined in \cite[Eqn. (3.16) and Eqn.\ (6.1)]{FiMoPa2}.

\begin{thm}[{\cite{FiMoPa1}, \cite{FiMoPa2}}]
Let $\PB = \set{P(k)}_{k \in \R^d}$ be a family of projectors as in Question $(\mathrm{Q}_d)$. Assume that $1 \le d \le 3$. Then a global Bloch frame for $\PB$ satisfying \ref{item:F1}, \ref{item:F2} and \ref{item:F3} exists:
\begin{description}
\item[If $\PB$ satisfies (P$_{3,+}$)] always; \\[-3mm]
\item[If $\PB$ satisfies (P$_{3,-}$)]  according to the dimension: 
\begin{description}
\item[If $d=1$] always;
\item[If $d=2$] if and only if 
$$ %\begin{equation} \label{delta=0,2d} 
\delta(\PB) = 0 \in \Z_2.
$$ %\end{equation}
%where $\delta(\PB)$ is defined in \cite[Eqn.\ (3.16)]{FiMoPa2};
\item[If $d=3$] if and only if 
$$ %\begin{equation} \label{delta=0,3d} 
\delta_{1,0}(\PB) = \delta_{1,+}(\PB) = \delta_{2,+}(\PB) = \delta_{3,+}(\PB) = 0 \in \Z_2.
$$ %\end{equation}
%where $\delta_{1,0}(\PB)$, $\delta_{1,+}(\PB)$, $\delta_{2,+}(\PB)$ and $\delta_{3,+}(\PB)$ are defined in \cite[Eqn.\ (6.1)]{FiMoPa2}.
\end{description}
\end{description}
\end{thm}

From the above table of results, we see that a positive answer to Question (Q$_d$) is in general \emph{topologically obstructed}: in particular, new interesting topological obstructions appear in the case of fermionic time-reversal symmetries for $d=2, 3$. Moreover, the peculiarity of these obstruction is that they are \emph{$\Z_2$-valued}, in contrast with the case of non-TR-symmetric Bloch bundles, where the obstruction to the existence of a smooth $\tau$-equivariant Bloch frame is encoded in the $\Z$-valued Chern indices.

We briefly recall, for the reader's convenience, our definition of the $\Z_2$ invariant $\delta(\PB)$ for a $2$-dimensional family of projectors $\PB = \set{P(k)}_{k \in \R^2}$ satisfying \ref{item:smooth}, \ref{item:tau} and (P$_{3,-}$). 
The alternative definitions, appearing in \cite{FuKa, GrafPorta}, can be shown to be equivalent.  

Since time-reversal symmetry relates the point $k$ and $-k$ in $\R^2$, a $\tau$-equivariant and time-reversal symmetric Bloch frame is completely determined by the values it attains on the \emph{effective unit cell}
\[ \Bred := \set{k = \sum_{j=1}^{2} k_j b_j \in \B: k_1 \ge 0}.  \]
The general strategy is to pick any smooth Bloch frame $\Psi$ on $\Bred$ (whose existence is guaranteed by the fact that $\Bred$ is contractible) and try to symmetrize it in order to impose $\tau$-equivariance and time-reversal symmetry, thus obtaining a smooth symmetric frame $\Phi$ on $\Bred$. The extension of $\Phi$ to the whole $\R^d$ is then obtained by imposing the relevant symmetries.%Since both $\Phi(k)$ and $\Phi(k)$ are orthonormal bases in $\Ran P(k)$, they differ just by the action of a unitary matrix $U(k) = (U_{ab}(k)}_{1 \le a,b \le m}$, \ie
%\begin{equation} \label{U}
%\phi_b(k) = \sum_{a=1}^{m} \psi_a(k) U_{ab}(k), \quad k \in \Bred.
%\end{equation}
%Hence we can equivalently treat $U(k)$ as the unknown of our problem.

Inside the effective unit cell, a special r\^ole is played by points which are fixed by the composition of a lattice translation $t_\lambda(k) = k+\lambda$, $\lambda \in \Lambda$, and the inversion $c(k) = -k$: such points $k_\lambda$ are then given by $k_\lambda = \lambda/2$. The symmetries that we want to impose on $\Phi$ influence the possible values that it can attain at these points $k_\lambda$, as well as at the edges that connect them and constitute the boundary $\partial \Bred$. In \cite{FiMoPa2}, it is shown that, starting from $\Psi(k)$, it is always possible to construct a unitary-matrix-valued map $\widehat{U} \colon \partial \Bred \to \U(\C^m)$ which is smooth and such that
\[ \widehat{\phi}_b(k) := \sum_{a=1}^{m} \psi_a(k) \, \widehat{U}_{ab}(k), \quad k \in \partial \Bred, \]
defines, on the boundary $\partial \Bred$, a smooth Bloch frame $\widehat{\Phi}(k)$ for $\PB$, which indeed satisfies all the relevant symmetries on $\partial \Bred$. Moreover, the definition of the frame $\widehat{\Phi}$ can be smoothly extended to the interior of $\Bred$ if and only if the map $\det \widehat{U} \colon \partial \Bred \approx S^1 \to U(1) \approx S^1$ has even winding number, \ie
\[ \deg([ \det \widehat{U}]) := \frac{1}{2\pi \iu} \oint_{\partial \Bred} \di z \, \partial_z \log \det \widehat{U}(z) \equiv 0 \mod 2 \]
(see \cite[Thm.\ 4]{FiMoPa2}). The topological obstruction to the existence of a smooth Bloch frame for $\PB$ which is both $\tau$-equivariant and time-reversal symmetric is thus encoded in the invariant
$$ 
\delta(\PB) := \deg([ \det \widehat{U}]) \bmod 2.
$$
%compare \cite[Eqn.\ (3.16)]{FiMoPa2}.
Finally, one shows that the value of the index $\delta(\PB)$ does not depend on the choice of the input frame $\Psi$, 
nor on the intermediate choices which are needed to construct $\widehat U$. Moreover, $\delta(\PB)$ defines a topological invariant of the family of projectors $\PB$. 

\bigskip

The general paradigm that one can extrapolate from this result is that the addition of symmetries (in this case, a fermionic time-reversal symmetry) refines the geometric structure, and leads to the emergence of new interesting topological invariants which label the phases of quantum matter (compare \cite{Kitaev,RyuSchnyder2010}). 
A point that should be emphasized is that the presence of symmetries cannot be always implemented by taking suitable quotients. Indeed, in the case of the $\Z^d$-symmetry, one can proceed by taking the quotient of the crystal momentum space $\R^d$ by the group of lattice translation $\Gamma^*$, thus obtaining the (Brillouin) torus $\T^d_*$, because the corresponding action of $\Gamma^*$ has no fixed points. The same quotient procedure can be also performed in the fibre Hilbert space $\Hf$, leading to the Bloch bundle, as was detailed above. However, one cannot take the subsequent quotient of the torus $\T^d_*$ by the inversion symmetry $c(k) = -k$, because the corresponding $\Z_2$-action this time \emph{has} fixed points, namely the points $k_\lambda = \lambda/2$. This would lead to a singular quotient $\T^d_* / \Z_2$, which could be described in terms of the $C^*$-algebraic methods of non-commutative geometry. In our approach, we prefer instead to impose the symmetries at the level of (global) Bloch frames, and study the existence of symmetric frames from the point of view of obstruction theory and differential geometry. In this way, putting on the special glasses of symmetries, we become able to detect finer details of the system under scrutiny. Granting us this power, Weyl's angel of topology has won its fight.

%%%%%%%%%%%%%%%%%%%%%%%%%%%%%%%%%%%%%%%%%%%%%%%%%%%%%%%%%%%%%%%

%\bigskip %%%%%%%%  AFFILIATIONS %%%%%%%%%%%%%%%%%%%%%%%%%%%
%\bigskip

{\footnotesize  %%%%%%%%%%%%%%%%%%%%%%%%%%%%%%%%%%%%

\vspace{15 mm}

\begin{tabular}{ll}
(D. Monaco) &  \textsc{SISSA -- International School for Advanced Studies}\\
&Via Bonomea 265, 34136 Trieste, Italy \\
&{E-mail address}: \href{mailto:dmonaco@sissa.it}{\texttt{dmonaco@sissa.it}} \\
\\
(G. Panati) & \textsc{Dipartimento di Matematica, \virg{La Sapienza} Universit\`{a} di Roma} \\
 &  Piazzale Aldo Moro 2, 00185 Rome, Italy \\
 &  {E-mail addresses}: \href{mailto:panati@mat.uniroma1.it}{\texttt{panati@mat.uniroma1.it}}, 
 \href{mailto:panati@sissa.it}{\texttt{panati@sissa.it}}  \\

\end{tabular}

\bigskip

}%%%Endfootnotesize  %%%%%%%%%%%%%%%%%%%%%%%%%%%%%%%%%%%
\end{document}